\documentclass[journal,draft,onecolumn,12pt]{IEEEtran}
\usepackage{amsfonts,color,morefloats}
\usepackage{amssymb,amsmath,latexsym,amsthm}

\newcommand{\rank}{{\mathrm{Rank}}}

\newcommand{\gf}{{\mathrm{GF}}}

\newcommand{\C}{{\mathcal{C}}}

\newcommand{\bc}{{\mathbf{c}}}
\newcommand{\bg}{{\mathbf{g}}}
\newcommand{\bu}{{\mathbf{u}}}

\newcommand{\bone}{{\mathbf{1}}}
\newcommand{\m}{\mathbb{M}}

\newtheorem{theorem}{Theorem}
\newtheorem{lemma}[theorem]{Lemma}

\newtheorem{corollary}[theorem]{Corollary}
\newtheorem{conjecture}{Conjecture}
\newtheorem{problem}{Open Problem}
\newtheorem{definition}{Definition}
\newtheorem{example}{Example}
\newtheorem{remark}{Remark}

\begin{document}

\title{Extended codes and deep holes of MDS codes \thanks{
Y. Wu's research was sponsored by the National Natural Science Foundation of China under Grant Number 12101326 and 62372247 and the
Natural Science Foundation of Jiangsu Province under Grant Number BK20210575. 
C. Ding's research was supported by The Hong Kong Research Grants Council, Proj. No. $16301523$. 
}
}

\author{Yansheng Wu,\thanks{Y, Wu is with the School of Computer Science, Nanjing University of Posts and Telecommunications, Nanjing
210023, China. Email: 
yanshengwu@njupt.edu.cn} 
\and Cunsheng Ding,\thanks{C. Ding is with the Department of Computer Science
                           and Engineering, The Hong Kong University of Science and Technology,
Clear Water Bay, Kowloon, Hong Kong, China. Email: cding@ust.hk}
\and Tingfang Chen\thanks{T. Chen is with the Department of Computer Science
                           and Engineering, The Hong Kong University of Science and Technology,
Clear Water Bay, Kowloon, Hong Kong, China. Email: tchenba@ust.hk}

}

\date{\today}
\maketitle

\begin{abstract} For a given linear code $\C$ of length $n$ over $\gf(q)$ and a nonzero vector $\bu$ in $\gf(q)^n$, Sun, Ding and Chen  defined an extended linear code $\overline{\C}(\bu)$ of $\C$, which is a generalisation of the classical extended 
code $\overline{\C}(-\bone)$ of $\C$ and called the second kind of an extended code of $\C$ (see arXiv:2307.04076 and arXiv:2307.08053). They developed some general theory of the extended codes   $\overline{\C}(\bu)$   and 
studied  the extended codes  $\overline{\C}(\bu)$ of several families of linear codes, including cyclic codes, projective two-weight codes, nonbinary Hamming codes, and a family of  reversible MDS cyclic codes. The objective of this paper is to investigate the extended codes  $\overline{\C}(\bu)$ of MDS codes $\C$ over finite fields. The main result 
of this paper is that the extended code $\overline{\C}(\bu)$ of an MDS $[n,k]$ code $\C$ remains MDS if and only if the covering radius $\rho(\mathcal{C}^{\bot})=k$ and the vector $\bu$ is a deep hole of the dual code $\C^\perp$.  As applications of this main result, the extended codes of the GRS codes and extended GRS codes are investigated 
and the covering radii of several families of MDS codes are determined. 
\end{abstract}


\section{Introduction}\label{sec-intro} 

 Let $n$ be a positive integer and $q$ be a power of a prime. 
An $[n,k,d]$ linear code $\C$ over $\gf(q)$ is a $k$-dimensional subspace of $\gf(q)^n$ with minimum 
Hamming distance $d$. The minimum distance $d$ of
a linear code $\mathcal{C}$ is upper bounded by $d\le n-k+1$, which is called the {\em Singleton bound}.
If $d= n-k+1$, then $\mathcal{C}$ is called a {\em maximum distance separable} (MDS) code. Namely, an MDS code has the greatest error correcting capability
 when its length and dimension are fixed.  MDS codes have wide applications in communication systems, cryptography  ad data storage systems. There are many known constructions for MDS codes, including 
 the {\em generalized Reed-Solomon} (GRS) codes and extended GRS codes \cite{RS}, 
 and those from arcs in projective geometry \cite{MS}, circulant matrices \cite{RL}, and Hankel matrices \cite{RS2}.
 
In coding theory, there are several methods to construct new linear codes with interesting parameters and properties from old ones. Among them are the shortening, puncturing,  augmenting and extending techniques. There are different ways to extend an $[n, k,d]$ linear code $\C$ over $\gf(q)$ into an $[n+1, k,d']$ linear code $\C$ over $\gf(q)$. Two of them are outlined below. 

For a given $[n, k,d]$ linear code $\C$ over $\gf(q)$, the first extending method goes as follows. 
Select a $k\times n$ generator matrix $G$ of $\C$ and a $k \times 1 $ vector $\bg$ in $\gf(q)^k$. 
Then construct a $k\times (n+1)$ matrix 
\begin{eqnarray}\label{eqn-tildeGv}
\widetilde{G}(G, \bg)=(G|\bg). 
\end{eqnarray} 
Let $\widetilde{\C}(G, \bg)$ be the linear code with generator matrix $\widetilde{G}(G,\bg)$. By definition, 
the extended code $\widetilde{\C}(G, \bg)$ has parameters $[n+1, k,\tilde{d}]$, where $\tilde{d}=d$ or 
$\tilde{d}=d+1$. It is easy to notice the following:
\begin{itemize}
\item The exact value of $\tilde{d}$ depends on both $G$ and $\bg$. 
\item The weight distribution of the extended code $\widetilde{\C}(G, \bg)$ depends on both $G$ and $\bg$. 
\item Two extended codes $\widetilde{\C}(G_1, \bg)$ and $\widetilde{\C}(G_2, \bg)$ may not be equivalent 
for a fixed nonzero vector $\bg$, where $G_1$ and $G_2$ are two $k \times n$ generator matrices of $\C$.
\end{itemize} 
There are a huge number of choices of a generator matrix $G$ of $\C$ and $q^k-1$ choices of a nonzero  vector $\bg$. 
Hence, by this extending technique a given $[n, k,d]$ linear code $\C$ over $\gf(q)$ could be extended 
into a huge number of codes $\widetilde{\C}(G, \bg)$. Seroussi and Roth \cite{SR}  proved that a GRS code can be extended by this extending technique while preserving the MDS property if and only if the resulting extended code is  a GRS code or extended GRS code.
This extending technique is not well studied in the 
literature and has been applied only to MDS codes and a few other types of linear 
codes. The extended GRS codes and Roth-Lempel codes are such extended codes.  To distinguish these  
extended codes $\widetilde{\C}(G, \bg)$ from other types of extended codes of $\C$, we call those codes 
the \emph{first kind of extended codes}.

For a given $[n, k,d]$ linear code $\C$ over $\gf(q)$, the second extending method goes as follows \cite{SDC}. 
Let ${\bf u}=(u_1, u_2, \ldots, u_n)\in \gf(q)^n$ be any nonzero vector. Define an $[n+1,k, \overline{d}]$ code $\overline{\mathcal C}({\bf u})$ over $\gf(q)$ by 
\begin{equation}\label{eq1}
\overline{\mathcal C}({\bf u})=\bigg\{  (c_1, \ldots, c_n,c_{n+1}): (c_1, \ldots, c_n)\in \mathcal C,c_{n+1}=\sum_{i=1}^nu_ic_i\bigg\},
\end{equation}
where $\overline{d} = d$ or $\overline{d} = d +1$. All extended codes $\overline{\mathcal C} (a{\bf 1})$ with $a \neq 0$ are called the standardly extended codes and all the other extended codes $\overline{\mathcal C} ({\bf u})$ with ${\bf u} \neq a{\bf 1}$ for all $a \in \gf(q)^*$  are called nonstandardly extended codes.
It is easy to observe the following:
\begin{itemize}
\item If ${\bf u}\in \mathcal C^{\bot}$ (i.e., the dual of $\C$), then $c_{n+1}=0$ for all codewords 
$\bc=(c_1, \ldots, c_{n+1})$ of $ \overline{\C}(\bu)$. 
In this case, $\overline{\mathcal C}({\bf u})$ is not interesting and is said to be trivial. 
\item The extended code $\overline{\C}(\bu)$ depends on $\bu$ but does not require a selected generator matrix of $\C$. 
\item $\overline{\C}(\bu_1)$ and $\overline{\C}(\bu_2)$ may have different minimum distances or different weight distributions. 
\end{itemize} 
The classical extended codes $\overline{\C}(-\bone)$ of a number of types of linear codes were studied in the literature. Recently,  nonstandardly extended codes $\overline{\C}(\bu)$ of linear codes were investigated in \cite{SDC}. To distinguish the extended codes $\overline{\C}(\bu)$ of $\C$ from those extended codes 
$\widetilde{\C}(G, \bg)$, we called these  $\overline{\C}(\bu)$ the \emph{second kind of extended codes} 
of $\C$.

For a given $[n, k,d]$ linear code $\C$ over $\gf(q)$, the two kinds of extended codes $\overline{\C}(\bu)$ 
and $\widetilde{\C}(G, \bg)$ are different but related in the following senses. 
We have the following remarks. 
\begin{itemize}
\item For each $\overline{\C}(\bu)$, there is a vector $\bg$ such that $\overline{\C}(\bu)=\widetilde{\C}(G, \bg)$ 
for each generator matrix $G$. 
\item For a given  $\widetilde{\C}(G, \bg)$, where $G$ is a $k \times n$ generator matrix of $\C$, there are $q^{n-k}$ 
vectors $\bu$ in $\gf(q)^n$ such that $\overline{\C}(\bu)=\widetilde{\C}(G, \bg)$, as $\rank(G)=\rank(G|\bg)=k \leq n$. These $\bu$ are the solutions to the system of equations $G\bu^T=\bg$.  
\end{itemize}
Hence, both extending techniques yield the same set of $[n+1, k]$ linear codes when they are applied to an 
$[n, k]$ linear code $\C$.   

The first kind of extended codes of only a few families of MDS and other types of codes were studied in the literature. 
Relatively speaking, the second kind of extended codes of more families of linear codes were investigated in the 
literature (see, e.g., \cite{SD, SDC}). By definition, it looks harder to determine the minimum distance of an extended 
code $\widetilde{\C}(G, \bg)$ than the minimum distance of an extended code $\overline{\C}(\bu)$ in general. 
  
This paper is motivated by the following open problem:
  \begin{problem} \label{problem-3} {\rm
Is there a vector $\bu\in \gf(q)^n$ such that the extended code $\overline{\C}(\bu)$ remains  MDS for a given MDS code $\C$ of length $n$ over $\gf(q)$? What are such vectors $\bu\in \gf(q)^n$ if they exist? 
}
 \end{problem}
   
In this paper, we will investigate the second kind of extended codes $\overline{\C}(\bu)$ of MDS codes $\C$ regarding Open Problem \ref{problem-3}. The rest of the paper is arranged as follows. In Section \ref{sec-prel}, we will recall some required  concepts and results on the extended codes and (extended) GRS codes over finite fields. In Section \ref{sec-mainres}, we will establish a relation between the extended codes $\overline{\C}(\bu)$ of MDS codes and deep holes of the dual codes $\C^\perp$, which is the main result of this paper. As applications of this main result, we will investigate  the second kind of extended codes of the GRS codes and extended GRS codes and also give partial answers to the open problem above for these special families of MDS codes in Sections \ref{sec-special1} and \ref{sec-special2}, respectively. 
As another application of the main result, we will determine the covering radii of several families of MDS codes in Section \ref{sec-special3}. Finally, we will conclude this paper and make some concluding remarks  in Section \ref{sec-final}.

\section{Preliminaries}\label{sec-prel} 

In this section, we will recall some concepts and known results, which will be needed later. 

\subsection{A basic result of the extended codes $\overline{\C}(\bu)$ of linear codes $\C$}

The dual code of an $[n,k,d]$ linear code $\mathcal C$ over $\gf(q)$ is defined by
$$\mathcal C^{\bot} = \{ (x_1, \ldots, x_{n})={\bf x}\in \gf(q)^{n} \mid \langle {\bf x},{\bf y}\rangle=\sum_{i=1}^{n}x_iy_i=0~  \forall~  {\bf y}= (y_1, \ldots,y_{n})\in \mathcal C\}.$$

Recall the definition of the extended code $\overline{\C}(u)$ in \eqref{eq1}.  It is easy to check that the following lemma holds.

\begin{lemma}\label{lem1}{\rm  \cite{SDC} ~Let $\mathcal C$ be an $[n,k,d]$ linear code over $\gf(q)$ and ${\bf u}\in \gf(q)^n$. If $\mathcal C $ has generator matrix $G$ and parity check matrix $H$, then the generator and parity check matrices for the extended code $\overline{\mathcal C}({\bf u})$ in \eqref{eq1} are $(G~ G{\bf u}^T)$ and \begin{equation*}\left(\begin{array}{cccc} H & {\bf 0}^T\\ {\bf u} & -1
\end{array}\right), \end{equation*}
where ${\bf u}^T$ denotes the transpose of ${\bf u}$.
}
\end{lemma}

\subsection {The GRS codes and extended GRS codes}

In this section, we recall the GRS codes and the extended GRS codes as well as their basic properties.  

\begin{definition}\label{def1} {\rm
For a positive integer $k$, let $$ \gf(q)[x]_k=\{f(x)\in \gf(q)[x]|\mbox{ deg}(f(x))\le k-1\},$$
which is a $ \gf(q)$-linear space of dimension $k$. Take ${\bf a}=(\alpha_1,\alpha_2,\ldots, \alpha_n)\in \gf(q)^n$, where  $\alpha_1,\ldots, \alpha_n$ are pairwise distinct,   ${\bf v}=(v_1,v_2,\ldots,v_n)\in ( \gf(q)^*)^n$, 
and   $k\le n \le q$.  The {\it generalized Reed-Solomon code}
(GRS code for short) is defined by 
$$\mathcal{C}_k({\bf a}, {\bf v})=\{(v_1f(\alpha_1), v_2f(\alpha_2),\ldots, v_nf(\alpha_n))|
\ f(x)\in  \gf(q)[x]_k\}.$$}
\end{definition}

It is well known that 
$\mathcal{C}_k({\bf a}, {\bf v})$  is an $[n, k, n-k+1]$  MDS code over $ \gf(q)$. When $v_i=1$ for all $1\le i\le n,$ the code $\mathcal{C}_k({\bf a}, {\bf 1})$ is called a \emph{Reed-Solomon code}.
It is easily seen that   $\mathcal{C}_k({\bf a}, {\bf v})$ has  generator matrix 
\begin{equation}\label{eq2}
G_k=\left(\begin{array}{cccc} v_1\alpha_1^0 &v_2\alpha_2^0 &\ldots  &v_n\alpha^0_{n}\\ v_1\alpha_1^1 &v_2\alpha_2^1 &\ldots & v_n\alpha^1_{n}\\ \vdots &\vdots &\ddots&\vdots \\ v_1\alpha_1^{k-1} &v_2\alpha_2^{k-1} &\ldots & v_n\alpha^{k-1}_{n}\end{array}\right) \end{equation}   and parity check matrix  \begin{equation}\label{eq3}
H_{n-k}=\left(\begin{array}{cccc} w_1\alpha_1^0 &w_2\alpha_2^0 &\ldots  &w_n\alpha^0_{n}\\ w_1\alpha_1^1 &w_2\alpha_2^1 &\ldots & w_n\alpha^1_{n}\\ \vdots &\vdots&\ddots &\vdots \\ w_1\alpha_1^{n-k-1} &w_2\alpha_2^{n-k-1} &\ldots & w_n\alpha^{n-k-1}_{n}\end{array}\right),\end{equation}
where these $w_i$ are given in the next lemma.  

\begin{lemma} \label{lem2} {\rm  \cite{HP03} $\mathcal{C}_k^\bot({\bf a}, {\bf v})=\mathcal{C}_{n-k}({\bf a}, {\bf w})$, where  ${\bf w}=(w_1,w_2,\ldots, w_n)$ and for $i=1,2,\ldots,n$, \begin{equation}\label{eq4}
 {w_i} = \frac{1}{v_i\prod_{1\leq j\leq n,j\neq i}(\alpha_i-\alpha_j)}.\end{equation}

}
\end{lemma}

Lemma \ref{lem2} tells us that the dual  code of  $\mathcal{C}_k^\bot({\bf a}, {\bf v})$ is also a GRS code with parameters $[n, n-k,k+1]$.

\begin{definition}\label{def2} {\rm
By including the point at infinity $\infty=(0, 0, \ldots, 0, 1) \in \gf(q)^k$, the {\it extended GRS code} (EGRS code for short) is defined as follows:
$$\mathcal{C}_k({\bf a}, {\bf v},\infty)=\{(v_1f(\alpha_1), v_2f(\alpha_2),\ldots, v_{n}f(\alpha_{n}),f_{k-1})|
\mbox{ for all }f(x)\in \gf(q)[x]_k\},$$ where $f_{k-1}$ is the coefficient of $x^{k-1}$ in $f(x)$.
}
\end{definition}

It is well known that $\mathcal{C}_k({\bf a}, {\bf v},\infty)$ is an $[n+1, k, n-k+2]$  MDS code over $\gf(q)$. Furthermore, 
$\mathcal{C}_k({\bf a}, {\bf v},\infty)$ has generator matrix $G_{k,\infty}$ and parity check matrix $H_{n+1-k, \infty}$, where 
 \begin{equation}\label{eq5}
G_{k,\infty}:=(G_k| \infty^T)=\left(\begin{array}{ccccc} v_1\alpha_1^0 &v_2\alpha_2^0 &\ldots  &v_{n}\alpha_n^0&0\\ v_1\alpha_1^1 &v_2\alpha_2^1 &\ldots & v_{n}\alpha^1_{n}&0\\ \vdots &\vdots &\ddots&\vdots& \vdots\\ v_1\alpha_1^{k-2} &v_2\alpha_2^{k-2}&\ldots & v_{n}\alpha_{n}^{k-2}&0\\v_1\alpha_1^{k-1} &v_2\alpha_2^{k-1}&\ldots & v_{n}\alpha_{n}^{k-1}&1\end{array}\right) \end{equation}
and  \begin{equation}\label{eq6}
H_{n+1-k, \infty}:=\left(\begin{array}{ccccc} w_1\alpha_1^0 &w_2\alpha_2^0 &\ldots  &w_n\alpha^0_{n}&0\\ w_1\alpha_1^1 &w_2\alpha_2^1 &\ldots & w_n\alpha^1_{n}&0\\ \vdots &\vdots &\ddots&\vdots&\vdots \\ w_1\alpha_1^{n-k-1} &w_2\alpha_2^{n-k-1} &\ldots & w_n\alpha^{n-k-1}_{n}&0\\ w_1\alpha_1^{n-k} &w_2\alpha_2^{n-k} &\ldots & w_n\alpha^{n-k}_{n}&-1\end{array}\right),\end{equation} 
where $G_k$ was defined in \eqref{eq2}. 
Using our earlier notation, we have 
$$
\mathcal{C}_k({\bf a}, {\bf v},\infty)=\widetilde{ \mathcal{C}_k({\bf a}, {\bf v})} (G_k, \infty^T). 
$$
Hence, $\mathcal{C}_k({\bf a}, {\bf v},\infty)$ is a first kind of extended code of $\mathcal{C}_k({\bf a}, {\bf v})$. 
The following result is well known. 

\begin{lemma}\label{lem3}{\rm  \cite[Lemma 2.1]{SYY}
   $\mathcal{C}_k^{\bot}({\bf a}, {\bf v},\infty)=\mathcal{C}_{n+1-k}({\bf a}, {\bf w},\infty)$, where ${\bf w}=(w_1,w_2,\ldots, w_n)$ and ${w_i} $ was given in \eqref{eq4} for $i=1,2,\ldots,n$.
}
\end{lemma}

By Lemma \ref{lem3}, $\mathcal{C}_k^{\bot}({\bf a}, {\bf v},\infty)$ is also an EGRS code with parameters $[n+1, n+1-k,k+1]$.


\section{The extended codes and deep holes of MDS codes}\label{sec-mainres}

In this section, we will establish the relationship between extended codes and deep holes of MDS codes. 
Before doing this, we need to recall the definition of deep holes of a linear code. 
 
The \emph{covering radius}  of a code $\mathcal C$, denoted by $\rho(\mathcal C)$, is the maximum distance from any vector in $\gf(q)^n$ to the nearest codeword in $\C$. A \emph{deep hole} is a vector achieving the covering radius. Before presenting our main result, we need to recall the following well-known characterisation  of covering radius of a linear code in terms of a generator matrix. 

\begin{lemma}{\rm \label{lem4} \cite[Lemma II.7]{ZWK} Let $G$  denote a generator  matrix of an $[n,k]$ MDS code $\mathcal C$. Then the covering radius  $\rho(\mathcal C)=n-k$ if and only if there exists a vector ${\bf x}\in \gf(q)^n$ such that the $(k+1)\times n$ matrix $\bigg(\frac{G}{{\bf x}}\bigg)$ generates an $[n,k+1]$ MDS code.



}
\end{lemma}

The determination of deep holes for a given linear code is a hard problem in general.  Below we provide a general characterization of deep holes of MDS codes with parity check matrices. The following lemma is a generalisation of \cite[Proposition 2]{ZCL}, which was stated only for the Reed-Solomon codes.  

\begin{lemma}{\rm \label{lem5}    

 (1)  \cite[Lemma II. 2]{ZWK} Let $\mathcal C$ be an  $[ n, k ]$ linear code with parity check matrix $H$. The vector ${\bf u}$ is a deep hole of $\mathcal C$ if and only if the vector $H{\bf u}^T$ can not be written as a linear combination of any $(\rho(\mathcal C) -1)$ columns of $H$.
 
 (2)   Suppose $G$ is a generator matrix of an $[n,k]$ MDS code $\mathcal C $ over $\gf(q)$  with covering radius $\rho = n - k$. Then $\bu \in \gf(q)^n$  is a deep hole of $\mathcal C$ if and only if
$$G'=\bigg( \frac{G}{\bu}\bigg)$$ generates an MDS code.
}
\end{lemma}
\begin{proof} (2). ($\Longrightarrow$). Suppose that $\bu$ is a deep hole of $\mathcal C$, we need to show that $G'$ is a generator matrix for another MDS code. Equivalently, we need to show that any $k+1$ columns of $G'$ are linearly independent.

Assume there exist $k+1$ columns of $G'$ which are linearly dependent. Without loss of generality, assume that the first $k+1$ columns of $G'$ are linearly dependent. Consider the submatrix consisting of the intersection of the first $k+1$ rows and $k+1$ columns of $G'$. Hence there exist $a_1,\ldots, a_k\in \gf(q)$, not all zero, such that $$( u_1,\ldots, { u}_{k+1})=a_1{\bf r}_{1,k+1}+\cdots+a_k{\bf r}_{k,k+1},$$where $\bu=(u_1, \ldots, u_n)$ and ${\bf r}_{i,k+1} $ is the vector consisting of the first $k+1$ elements of $i$-th row of $G$ for $1\le i\le k$. Let $\bg=a_1{\bf r}_1+\cdots+a_k{\bf r}_k\in \mathcal C$, where ${\bf r}_i$ is the $i$-th row of $G$ for $1\le i\le k $. We have $$d(\bu,\bg)\le n-(k+1)<\rho,$$ which is a contradiction with the assumption that $\bu$ is a deep hole of $\mathcal C$.

($\Longleftarrow$). Now suppose $G'$ is a generator matrix for an MDS code, i.e., any $k+1$ columns of $G'$ are linearly independent. We need to show that $d(\bu, \mathcal C)=n-k$. 

Assume that $d(\bu, \mathcal C)<n-k$. Equivalently, there exist $a_1, \ldots, a_k\in \gf(q)$ such that $\bu $ and $\bg=a_1{\bf r}_1+\cdots+a_k{\bf r}_k\in \mathcal C$ have more than  $k$ common coordinates, where ${\bf r}_i$ is the $i$-th row of $G$ for $1\le i\le k$. Without loss of generality, assume that the first $k+1$ coordinates of $\bu$ and $\bg$ are the same. Consider the submatrix consisting of the first $k+1$ columns. Since the rank of the matrix is less than $k+1$, thus the first $k+1$ columns of $G'$ are linearly dependent, which contradicts the assumption. 
This completes the proof.
\end{proof}

We are now ready to present our main result, which gives an answer to Open Problem \ref{problem-3} regarding MDS codes.

\begin{theorem} \label {thm6}{\rm  Let $\mathcal C$ be an $[n,k]$ MDS code over $\gf(q)$.  Then for any  ${\bf u}\in \gf(q)^n$, the extended code  $\overline{\mathcal C}({\bf u})$ in \eqref{eq1} is MDS   if and only if $\rho(\mathcal C^{\bot})=k$ and  ${\bf u}$ is a deep hole of the dual code $\mathcal C^{\bot}$.

}
\end{theorem}

\begin{proof} Suppose that the code $\mathcal C$ has  generator  and parity check matrices $G$ and   $H$, respectively. 
By Lemma \ref{lem1},   the generator and parity check matrices for the code $\overline{\mathcal C}({\bf u})$ are $(G~ G{\bf u}^T)$ and \begin{equation*}\left(\begin{array}{cccc} H & {\bf 0}^T\\ {\bf u} & -1
\end{array}\right). \end{equation*} 
 Then the code $\overline{\mathcal C}({\bf u})$ is MDS if and only if the following hold:

\begin{enumerate}
\item  The submatrix  consisting of the vector $(0,\ldots,0,-1)^T$ and any $n-k$ columns from $\bigg( \frac{H}{\bf u}\bigg)$ is nonsingular.  

\item The submatrix consisting of any $n+1-k$ columns from $\bigg( \frac{H}{\bf u}\bigg)$ is nonsingular.

\end{enumerate}

($\Longrightarrow$). Assume that $\overline{\mathcal C}({\bf u})$ is MDS with parameters $[n+1, k,n+2-k]$.
Since $\mathcal C$ is an $[n,k, n-k+1]$ MDS code, the condition 1) holds.
Let $\mathcal D$ be the code generated by $\bigg( \frac{H}{\bf u}\bigg)$.  From 2),  the dual code of $\mathcal D$ has parameters $[n, k-1, n-k+2]$ and hence $\mathcal D$ is also MDS with parameters $[n,n+1-k,k]$.  
From Lemma \ref{lem4}, $\rho(\mathcal C^{\bot})=n-(n-k)=k$. From Lemma \ref{lem5}, ${\bf u}$ is a deep hole of $\mathcal C^{\bot}$ if and only if the vector $G{\bf u}^T$ cannot be written as a linear combination of any $(\rho(\mathcal C^{\bot}) -1)=(k-1)$ columns of $G$.
Note that $(G~ G{\bf u}^T)$ is a parity check matrix of $(\overline{\mathcal C}({\bf u}))^{\bot}$, which is an $[n+1, n+1-k,k+1]$ MDS code. Hence the submatrix consisting of any $k$ columns from $(G~G{\bf u}^T)$ is nonsingular. Hence ${\bf u}$ is a deep hole of $\mathcal C^{\bot}$.

($\Longleftarrow$).  Let us assume that $\rho(\mathcal{C}^\bot) = k$, and ${\bf u}$ represents a deep hole of the dual code $\mathcal{C}^\bot$. According to Lemma \ref{lem5}(2), it follows that $\left(\frac{H}{{\bf u}}\right)$ generates an MDS code with parameters $[n, n+1-k, k]$. Furthermore, the dual code of this generated code has parameters $[n, k-1, n+2-k]$. Therefore, the condition 2) mentioned above is satisfied.

Moreover, the condition 1) is also fulfilled because $\mathcal{C}$ is an $[n, k, n-k+1]$ MDS code, and $H$ serves as a parity-check matrix for $\mathcal{C}$. Consequently, the code $\overline{\mathcal{C}}({\bf u})$ is proven to be MDS. This completes the proof. 
\end{proof}

As applications of Theorem \ref{thm6}, we will deal with the extended codes of the GRS codes and EGRS codes in 
Sections \ref{sec-special1} and \ref{sec-special2}, and will determine the covering radii of some MDS codes in Section
\ref{sec-special3}. 




\section{Applications of Theorem \ref{thm6} to the generalized Reed-Solomon 
codes}\label{sec-special1}

In this section, we consider two applications of Theorem \ref{thm6} to the generalized Reed-Solomon 
codes $\mathcal{C}_k({\bf a}, {\bf v})$ and their related codes.


\subsection{The covering radius of $\mathcal{C}_k({\bf a}, {\bf v})^\perp$}

As an application of Theorem \ref{thm6}, we derive the covering radius of $\mathcal{C}_k({\bf a}, {\bf v})^\perp$ in this subsection. 
It was shown earlier that $\mathcal{C}_k({\bf a}, {\bf v},\infty)=\widetilde{\mathcal{C}_k({\bf a}, {\bf v})}(G_k, \infty^T)$ with $\infty^T=(0,
\ldots, 0,1)^T$ in Equation \eqref{eqn-tildeGv}. We would first find a vector $\bu \in \gf(q)^n$ such that 
\begin{equation}\label{eqn-ding1}
\mathcal{C}_k({\bf a}, {\bf v},\infty)=\overline{\mathcal{C}_k({\bf a}, {\bf v})}({\bf u}). 
\end{equation} 
The next theorem documents such a vector $\bu$.

\begin{theorem} \label{thm7}{\rm Take ${\bf a}=(\alpha_1,\alpha_2,\ldots, \alpha_n)\in \gf(q)^n$, where  $\alpha_1,\alpha_2,\ldots, \alpha_n$ are pairwise distinct,   ${\bf v}=(v_1,v_2,\ldots,v_n)\in ( \gf(q)^*)^n$, and   $k\le n \le q$. Then
$\mathcal{C}_k({\bf a}, {\bf v},\infty)=\overline{\mathcal{C}_k({\bf a}, {\bf v})}({\bf u})$, where ${\bf u}=(u_1, u_2, \ldots, u_n)\in \gf(q)^n$ and  $$u_i=   \frac{\alpha_i^{n-k}}{v_i\prod_{1\leq j\leq n,j\neq i}(\alpha_i-\alpha_j)}.$$

}
\end{theorem}

\begin{proof}

It is clear that $$G_{n}=\left(\begin{array}{cccc} v_1\alpha_1^0 &v_2\alpha_2^0 &\ldots  &v_n\alpha^0_{n}\\ v_1\alpha_1^1 &v_2\alpha_2^1 &\ldots & v_n\alpha^1_{n}\\ \vdots &\vdots &\ddots&\vdots \\ v_1\alpha_1^{n-1} &v_2\alpha_2^{n-1} &\ldots & v_n\alpha^{n-1}_{n}\end{array}\right)$$ is invertible.
Then according to the Cramer Rule, the system of the equations
\begin{equation}\label{eq7}
\left(\begin{array}{cccc} v_1\alpha_1^0 &v_2\alpha_2^0 &\ldots  &v_n\alpha^0_{n}\\ v_1\alpha_1^1 &v_2\alpha_2^1 &\ldots & v_n\alpha^1_{n}\\ \vdots &\vdots &\ddots&\vdots \\ v_1\alpha_1^{n-1} &v_2\alpha_2^{n-1} &\ldots & v_n\alpha^{n-1}_{n}\end{array}\right)\left(\begin{array}{c}x_1\\x_2\\ \vdots
\\ x_n\end{array}\right)=\left(\begin{array}{c}0\\ \vdots
\\ 0\\1 \end{array}\right)
\end{equation}
has a nonzero solution. Hence the system $G_{n-1}X=0$ has a nonzero solution $(w_1,\ldots, w_{n})^T$ where ${w_i} = \frac{1}{v_i\prod_{1\leq j\leq n,j\neq i}(\alpha_i-\alpha_j)}$ was given in \eqref{eq4} for $i=1,2,\ldots,n$. 
Therefore, we have 
$$G_k{\bf u}^T=\left(\begin{array}{cccc} v_1\alpha_1^0 &v_2\alpha_2^0 &\ldots  &v_n\alpha^0_{n}\\ v_1\alpha_1^1 &v_2\alpha_2^1 &\ldots & v_n\alpha^1_{n}\\ \vdots &\vdots &\ddots&\vdots \\ v_1\alpha_1^{k-1} &v_2\alpha_2^{k-1} &\ldots & v_n\alpha^{k-1}_{n}\end{array}\right)\left(\begin{array}{c}u_1\\u_2\\ \vdots
\\ u_n\end{array}\right)=\left(\begin{array}{c}0\\ \vdots
\\ 0\\1\end{array}\right).$$

Notice that the code $\mathcal{C}_k({\bf a}, {\bf v})$ has generator matrix $G_k$ in \eqref{eq2}. By Lemma \ref{lem1} and the system of equations above,  $\overline{\mathcal{C}_k({\bf a}, {\bf v})}({\bf u})$ has generator matrix $(G_k ~ G_k{\bf u}^T)=G_{k,\infty}$, which is exactly  a generator matrix of $\mathcal{C}_k({\bf a}, {\bf v},\infty)$ in \eqref{eq5}. 
This completes the proof.
\end{proof}

Recall that  $\mathcal{C}_k({\bf a}, {\bf v},\infty)$ is MDS. Combining Theorems \ref{thm6} and \ref{thm7} then gives the following result. 

\begin{theorem} \label{thm67}{\rm Let ${\bf a}=(\alpha_1,\alpha_2,\ldots, \alpha_n)\in \gf(q)^n$, where  $\alpha_1,\alpha_2,\ldots, \alpha_n$ are pairwise distinct,   ${\bf v}=(v_1,v_2,\ldots,v_n)\in ( \gf(q)^*)^n$, and   $k\le n \le q$. Then
$$
\rho( \mathcal{C}_k({\bf a}, {\bf v})^\perp)=k 
$$ 
and the $\bu$ given in Theorem \ref{thm7} is a deep hole of $ \mathcal{C}_k({\bf a}, {\bf v})^\perp$. 
}
\end{theorem}

\subsection{Finding all $u$'s  such that $\overline{\mathcal{C}_k({\bf a}, {\bf v})}({\bf u})$ is MDS in some cases}

In this subsection, we will find all vectors $u$  such that $\overline{\mathcal{C}_k({\bf a}, {\bf v})}({\bf u})$ is MDS in some cases. By Theorem \ref{thm6}, we have the following theorem immediately:

\begin{theorem}\label{thm8}{\rm Take ${\bf a}=(\alpha_1,\alpha_2,\ldots, \alpha_n)\in \gf(q)^n$, where  $\alpha_1,\alpha_2,\ldots, \alpha_n$ are pairwise distinct,   ${\bf v}=(v_1,v_2,\ldots,v_n)\in ( \gf(q)^*)^n$, and   $k\le n \le q$.  For a vector ${\bf u}=(u_1, u_2, \ldots, u_n)\in \gf(q)^n$, the extended code $\overline{\mathcal{C}_k({\bf a}, {\bf v})}({\bf u})$ in \eqref{eq1} is MDS with parameters  $[n+1, k, n-k+2]$ if and only if $\rho(\mathcal{C}_{n-k}({\bf a}, {\bf w}))=k$ and
${\bf u}$ is a deep hole of the code $\mathcal{C}_{n-k}({\bf a}, {\bf w})$, where  ${\bf w}=(w_1,w_2,\ldots, w_n)$ and for each $1\le i\le n$, $w_i$ was given in \eqref{eq4}.
}
\end{theorem} 

To find all ${\bf u}$ such that $\overline{\mathcal{C}_k({\bf a}, {\bf v})}({\bf u})$ is MDS in some cases,  
we need to recall some results on deep holes of the GRS codes. 

Note that  any GRS code is of the form $\C'= \{ cM : c \in \C \}$ where $\C$ is an $[ n, k ]$ Reed-Solomon code over 
$\gf(q)$ and $M$ is an invertible $n \times n$ diagonal matrix over $\gf(q) $. 
Clearly, the set of deep holes of $\C'$ is $\{ {\bf x}M: {\bf x} \mbox{ is a deep hole of } \C\}$. 
Therefore, for the problem of determining the deep holes of the GRS codes, 
it surfaces to treat only Reed-Solomon codes.

Given a vector ${\bf f} = (f_1, f_2, \ldots , f_n) \in \gf(q)^n$, we consider the following Lagrange interpolating polynomial
\begin{equation*}f(x)=\sum_{i=1}^nf_i\prod_{j=1, j\neq i}^n \frac{x-\alpha_j}{\alpha_i-\alpha_j}\in \gf(q)[x],
\end{equation*}
where ${\bf a}=(\alpha_1,\alpha_2,\ldots, \alpha_n)\in \gf({q})^n$. 
The Lagrange interpolating polynomial is the unique polynomial in $\gf(q)[x]$ of degree less than $n$ that satisfies 
$f(\alpha_i) = f_i $ with $ 1\le i\le n$. We say that a function $f(x)$ generates a vector ${\bf f} \in \gf(q)^n$  if ${\bf f} = (f(\alpha_1), f(\alpha_2), \ldots, f(\alpha_n))$.

We say that  two functions $f(x)$ and $g(x)$ are equivalent over $\gf(q)$ if and only if there exist $a\in \gf(q)^*$ and $h(x)$ with degree less than $\mbox{deg}(f(x))$ such that $g(x)=af(x)+h(x)$.

\begin{lemma}{\rm  \cite[Theorem 1]{K}\label{lem9} Let $q>2$ be a power of a prime, $k\ge (q-1)/2, $ and $k<n\le q.$ Take  ${\bf a}=(\alpha_1,\alpha_2,\ldots, \alpha_n)\in \gf(q)^n$, where  $\alpha_1,\alpha_2,\ldots, \alpha_n$ are pairwise distinct. The only deep holes of $\mathcal{C}_k({\bf a}, {\bf 1})$ are generated by functions which are equivalent to the following:
$$f(x)=x^k, f_{\pi}=\frac{1}{x-\pi},$$
where $\pi\in \gf(q)\backslash \{\alpha_1, \ldots, \alpha_n\}$.


}

\end{lemma}

By Lemma \ref{lem9}, we have the following corollary, which gives all vectors $\bu$ such that 
$\overline{\mathcal{C}_k({\bf a}, {\bf v})}({\bf u})$ is MDS in the special case that $ (q-1)/2\le n-k$ 
and $k<n\le q$, provided that $\rho(\mathcal{C}_{n-k}({\bf a}, {\bf w}))=k$. 

\begin{corollary}\label{cor10}  {\rm  Let $q$ be a power of a prime, $ (q-1)/2\le n-k$, and $k<n\le q$.
For a vector ${\bf u}=(u_1, u_2, \ldots, u_n)\in \gf(q)^n$, the extended code $\overline{\mathcal{C}_k({\bf a}, {\bf v})}({\bf u})$ in \eqref{eq1} is MDS with parameters  $[n+1, k, n-k+2]$ if and only if  $\rho(\mathcal{C}_{n-k}({\bf a}, {\bf w}))=k$ and ${\bf u}$ is generated by functions which are equivalent to the following:
$$f(x)=x^{n-k}\cdot(\sum_{i=1}^nw_i\prod_{j=1, j\neq i}^n \frac{x-\alpha_j}{\alpha_i-\alpha_j}), f_{\pi}=\frac{1}{x-\pi}\cdot(\sum_{i=1}^nw_i\prod_{j=1, j\neq i}^n \frac{x-\alpha_j}{\alpha_i-\alpha_j})$$
where $w_i$ was given in \eqref{eq4} and $\pi\in \gf(q)\backslash \{\alpha_1, \ldots, \alpha_n\}$. 

}
\end{corollary} 

\begin{remark}{\rm The determination of all deep-holes of the Reed-Solomon codes is an interesting and 
difficult  problem in coding theory and combinatorics. This problem has received significant attention 
in the literature (see, e.g., \cite{K,ZWK,ZCL}).  New results on deep holes of the Reed-Solomon codes can be 
applied to Theorem \ref{thm8} to make sure that the extended codes are MDS.

}
\end{remark}

\section{Applications of Theorem \ref{thm6} to the extended Reed-Solomon 
codes $\mathcal{C}_k({\bf a}, {\bf 1},\infty)$}\label{sec-special2}

In this section, we will consider the extended code $\overline{\mathcal{C}_k({\bf a}, {\bf 1},\infty)}({\bf u})$ 
and find some $\bu$'s such that $\overline{\mathcal{C}_k({\bf a}, {\bf 1},\infty)}({\bf u})$ is MDS. In addition, 
we will determine the covering radius of $\mathcal{C}_k({\bf a}, {\bf 1},\infty)^\perp$ in some cases.

\subsection{The Roth-Lempel codes}

In this subsection, we recall the Roth-Lempel codes, which are a class of MDS codes 
that are non-Reed-Solomon type under certain condition \cite{RL}.
 


\begin{definition}\label{defn2} {\rm \cite{RL} Let $n$ and $k$ be two integers such that  $4\le k+1\le n\le q$.  Let $\alpha_1, \ldots, \alpha_{n}$ be pairwise distinct elements of $\gf(q)$, $\delta\in \gf(q)$, and ${\bf a}=(\alpha_{1},\alpha_{2},\ldots,\alpha_{n})$. Then the $[n+2,k]$ {\em Roth-Lempel  code}
$RL({\bf a},k,\delta, n+2)$
over $\gf(q)$ is generated by the matrix 
\begin{equation}\label{eq8} 
\left( \begin{array}{cccccc}
1 & 1& \ldots &  1 & 0 &0\\
\alpha_1& \alpha_2& \ldots &  \alpha_{n} & 0 &0\\
\vdots& \vdots& \ldots & \vdots & \vdots &\vdots\\
\alpha_1^{k-3}& \alpha_2^{k-3}& \ldots &  \alpha_{n}^{k-3} & 0 &0\\
\alpha_1^{k-2}& \alpha_2^{k-2}& \ldots &  \alpha_{n}^{k-2} & 0 &1\\
\alpha_1^{k-1}& \alpha_2^{k-1}& \ldots &  \alpha_{n}^{k-1} & 1 &\delta\\
\end{array} \right).
\end{equation}
}
\end{definition}

The following lemma gives a sufficient and necessary condition for the Roth-Lempel  code to be MDS.

\begin{lemma}\label{lem11}{\rm \cite{RL}  The $[n+2,k]$  Roth-Lempel  code $RL({\bf a},k,\delta, n+2)$ 
with  $4\le k+1\le n\le q$
 over $\gf(q)$  is MDS  if and only if the set $\{\alpha_1, \ldots, \alpha_{n}\}$ is an $(n, k-1, \delta)$-set in $\gf(q)$, where a set $S\subseteq \gf(q)$ of size $n$ is called an $(n,k-1, \delta)$-set in $\gf(q)$ if no $(k-1)$ elements of $S$ sum to $\delta$.

}

\end{lemma}


\subsection{The covering radius of $\mathcal{C}_k({\bf a}, {\bf 1},\infty)^\perp$}

As an application of Theorem \ref{thm6}, we derive the covering radius of $\mathcal{C}_k({\bf a}, {\bf 1},\infty)^\perp$ 
under certain conditions in this subsection. 

It is easy to check that $RL({\bf a},k,\delta, n+2)=\widetilde{\mathcal{C}_k({\bf a}, {\bf 1},\infty)}(G_{k,\infty}, \bg)$ with $\bg=(0,\ldots, 0,1,\delta)^T$ in Equation \eqref{eqn-tildeGv}. 
We would find a vector $\bu$ such that 
$\overline{\mathcal{C}_k({\bf a}, {\bf 1},\infty)}(\bu)=\widetilde{\mathcal{C}_k({\bf a}, {\bf 1},\infty)}(G_{k,\infty}, \bg)$.

\begin{theorem}\label{thm12}{\rm  We have the equality  
$$RL({\bf a},k,\delta,n+2)=\overline{\mathcal{C}_k({\bf a}, {\bf 1},\infty)}({\bf u}),$$ 
where ${\bf u}=(u_1, u_2, \ldots, u_n, \delta-a)\in \gf(q)^{n+1}$, 
$a=\sum_{i=1}^n \alpha_i$ and  $$u_i=   \frac{\alpha_i^{n+1-k}}{\prod_{1\leq j\leq n,j\neq i}(\alpha_i-\alpha_j)}.$$  

}
\end{theorem}

\begin{proof}  By Lemma \ref{lem1}, it suffices  to check that $G_{k,\infty}{\bf u}^T=(0, \ldots, 0,1, \delta)^T$. From Equation \eqref{eq7}, we only need to check that  $(\alpha_1^{k-1}, \ldots, \alpha_n^{k-1},1){\bf u}^T=\delta-a+\sum_{i=1}^nw_i\alpha_i^n$, where $w_i$ was given in \eqref{eq4} with $v_i=1$.

Let $m(x)=\prod_{j=1}^n(x-\alpha_j)=x^n+\sum_{j=0}^{n-1}a_jx^j$.  Note that $m(\alpha_i)=0$ and hence $\alpha_i^n=-\sum_{j=0}^{n-1}a_j\alpha_i^j$ for $1\le i\le n$.  By $G_n(w_1, \ldots, w_n)^T=(0, \ldots, 0 ,1)^T$, we have $$\sum_{i=1}^nw_i\alpha_i^n=-\sum_{j=0}^{n-1}a_j(\sum_{i=1}^{n}w_i\alpha_i^j)=-a_{n-1}=\sum_{i=1}^n\alpha_i=a.$$

This completes the proof.
\end{proof}


Combining Lemma \ref{lem11} and Theorems \ref{thm6} and \ref{thm12}, we obtain the following result. 

\begin{theorem}\label{thm612}
{\rm 
If there is a $\delta \in \gf(q)$ such that $\{\alpha_1, \ldots, \alpha_n\}$ is an $(n, k-1, \delta)$-set in $\gf(q)$, 
then 
$$
\rho( \mathcal{C}_k({\bf a}, {\bf 1},\infty)^\perp)=k 
$$
and the vector $\bu$ given in Theorem \ref{thm12} is a deep hole of $\mathcal{C}_k({\bf a}, {\bf 1},\infty)^\perp$. 
}
\end{theorem} 


Whether Theorem \ref{thm612} can be used to determine the covering radius $\rho( \mathcal{C}_k({\bf a}, {\bf 1},\infty)^\perp)$ of the MDS code $ \mathcal{C}_k({\bf a}, {\bf 1},\infty)^\perp$ depends on if there is a $\delta \in \gf(q)$ such that $\{\alpha_1, \ldots, \alpha_n\}$ is an $(n, k-1, \delta)$-set in $\gf(q)$.  This is an interesting open 
problem.  The following examples show the existence of such $\delta$ and the usefulness of Theorem \ref{thm612} sometimes.  


\begin{example}\label{exa1}
{\rm 

Let $n=q=4$ and $k=3$. Then $\gf(4)$ is a $(4, 2, 0)$-set in $\gf(4)$ and
in fact $\mathcal{C}_3({\bf a}, {\bf 1},\infty)^\perp$ is a $[5,2,4]$ linear code over $\gf(4)$.
By Theorem \ref{thm612}, we have $\rho( \mathcal{C}_3({\bf a}, {\bf 1},\infty)^\perp)=3$, which is confirmed  by Magma.

}
\end{example}

\begin{example}\label{exa2}
{\rm Let $n=q=8$ and $k=3$. Then $\gf(8)$ is a $(8, 2, 0)$-set in $\gf(8)$ and
in fact $\mathcal{C}_3({\bf a}, {\bf 1},\infty)^\perp$ is a $[9,6,4]$ linear code over $\gf(8)$.
By Theorem \ref{thm612}, we have $\rho( \mathcal{C}_3({\bf a}, {\bf 1},\infty)^\perp)=3$, which is confirmed  by Magma.

}
\end{example}

\begin{example}\label{exa3}
{\rm Let $n=q=8$ and $k=4$. According to Magma,  there is no $\delta \in \gf(8)$ such that $\gf(8)$ is a $(8, 3, \delta)$-set in $\gf(8)$. Namely, for any $\delta \in \gf(8)$, we can always find three distinct elements of $\gf(8)$ whose sum is $\delta$.
 In fact $\mathcal{C}_4({\bf a}, {\bf 1},\infty)^\perp$ is a $[9,5,5]$ linear code over $\gf(8)$ with covering radius  3.
}
\end{example}

\subsection{Some $\bu$'s such that  $\overline{\mathcal{C}_k({\bf a}, {\bf 1},\infty)}(\bu)$ is MDS via deep holes}

In this subsection, we will find some $\bu$'s such that  $\overline{\mathcal{C}_k({\bf a}, {\bf 1},\infty)}(\bu)$ is MDS via deep holes. 
By Theorem \ref{thm6} and Lemma \ref{lem3}, we have the following corollary.

\begin{corollary}\label{cor13}{\rm Let  $k\le n \le q$ and ${\bf a}=(\alpha_1,\alpha_2,\ldots, \alpha_n)\in \gf(q)^n$, where  $\alpha_1,\alpha_2,\ldots, \alpha_n$ are pairwise distinct. 
  For a vector ${\bf u}=(u_1, u_2, \ldots,u_{n+1})\in \gf(q)^{n+1}$, the extended code $\overline{\mathcal{C}_k({\bf a}, {\bf 1},\infty)}({\bf u})$ is MDS with parameters  $[n+2, k, n+3-k]$ if and only if $\rho(\mathcal{C}_{n+1-k}({\bf a}, {\bf w},\infty))=k$ and  ${\bf u}$ is a deep hole of the code $\mathcal{C}_{n+1-k}({\bf a}, {\bf w},\infty)$, where  ${\bf w}=(w_1,w_2,\ldots, w_n)$ and for $i=1,2,\ldots,n$, $$w_i=\frac{1}{\prod_{1\leq j\leq n,j\neq i}(\alpha_i-\alpha_j)}.$$
}
\end{corollary}

We will find some deep holes of  $\mathcal{C}_{n+1-k}({\bf a}, {\bf w},\infty)$ by utilizing some  results on  deep holes of $\mathcal{C}_{n-k}({\bf a}, {\bf w})$ in Corollary \ref{cor10}. By Corollary \ref{cor13}, these deep holes of  $\mathcal{C}_{n+1-k}({\bf a}, {\bf w},\infty)$ give MDS codes $\overline{\mathcal{C}_k({\bf a}, {\bf 1},\infty)}({\bf u})$, 
provided that $\rho(\mathcal{C}_{n+1-k}({\bf a}, {\bf w},\infty))=k$.

\begin{theorem}{\rm  \label{thm14} 

Assume that $\rho(\mathcal{C}_{n+1-k}({\bf a}, {\bf w},\infty))=k$.

(1) Suppose that ${\bf u}$ is a deep hole of $\mathcal{C}_{n-k}({\bf a}, {\bf w})$ generated by functions which are equivalent to the following:
$$f(x)=x^{n-k}\cdot(\sum_{i=1}^nw_i\prod_{j=1, j\neq i}^n \frac{x-\alpha_j}{\alpha_i-\alpha_j}).$$   For   $\delta\in \gf(q)$, 
$$\bg=(\alpha_1f(\alpha_1), \ldots, \alpha_nf(\alpha_n),\delta),$$ is a deep hole of 
$\mathcal{C}_{n+1-k}({\bf a}, {\bf w},\infty)$  if and only if  $-\delta \notin S_{n+1-k}$, where 
\begin{eqnarray}\label{eqs1} S_{n+1-k}=\bigg\{\sum_{i\in I}\alpha_i: \forall I\subseteq \{1,\ldots,n\} \mbox{ and } |I|=n+1-k \bigg\}.\end{eqnarray}



(2)  Suppose that ${\bf u}$ is a deep hole of $\mathcal{C}_{n-k}({\bf a}, {\bf w})$ generated by functions which are equivalent to the following:
 $$f_{\pi}(x)=\frac{1}{x-\pi}\cdot(\sum_{i=1}^nw_i\prod_{j=1, j\neq i}^n \frac{x-\alpha_j}{\alpha_i-\alpha_j}).$$  For   $\delta\in \gf(q)$, 
$$\bg=(f_{\pi}(\alpha_1), \ldots, f_{\pi}(\alpha_n),\delta),$$ is a deep hole of 
$\mathcal{C}_{n+1-k}({\bf a}, {\bf w},\infty)$  if and only if  $\delta \notin T_{\pi,n+1-k}$, where
\begin{eqnarray}\label{eqs2} T_{\pi,n+1-k}=\bigg\{\frac{1}{\prod_{i\in I}(\pi-\alpha_{i})}: \forall I\subseteq \{1,\ldots,n\} \mbox{ and } |I|=n+1-k \bigg\}.\end{eqnarray}


}
\end{theorem}


\begin{proof}   

The code $\mathcal{C}_{n+1-k}({\bf a}, {\bf w}, \infty)$ has generator matrix $H_{n+1-k,\infty}$ in \eqref{eq6} for $v_i=1$ with $1\le i\le n$.
By Lemma \ref{lem5} (2), ${\bg}$ is a deep hole of $\mathcal{C}_{n+1-k}({\bf a}, {\bf w}, \infty)$ if and only if $\bigg( \frac{H_{n+1-k,\infty}}{\bg}\bigg)$ generates an MDS code.

(1) When ${\bg}=(\alpha_1f(\alpha_1), \ldots, \alpha_nf(\alpha_n),\delta)$,  it suffices to prove that the submatrix consisting of any $(n+2-k)$ columns of the following matrix is nonsingular: 
\begin{equation}\label{eqHg}
\bigg( \frac{H_{n+1-k,\infty}}{\bg}\bigg)=\left(\begin{array}{ccccc}
 w_1\alpha_1^0 &w_2\alpha_2^0 &\ldots  &w_n\alpha^0_{n}&0\\ 
 w_1\alpha_1^1 &w_2\alpha_2^1 &\ldots & w_n\alpha^1_{n}&0\\ 
 \vdots &\vdots &\ddots&\vdots&\vdots \\ 
 w_1\alpha_1^{n-k-1} &w_2\alpha_2^{n-k-1} &\ldots & w_n\alpha^{n-k-1}_{n}&0\\ w_1\alpha_1^{n-k} &w_2\alpha_2^{n-k} &\ldots & w_n\alpha^{n-k}_{n}&-1\\
 w_1\alpha_1^{n+1-k} &w_2\alpha_2^{n+1-k} &\ldots & w_n\alpha^{n+1-k}_{n}&\delta\\
\end{array}\right).\end{equation}

If those columns from \eqref{eqHg} do not containing the vector $(0,\ldots, 0, -1, \delta)^T$, then the result holds because each induced submatrix is  Vandermonde and $\alpha_i$ are distinct. Then we consider the other case, without loss of generality we need to compute the determinant of the matrix:

\begin{equation*}\label{eq10}
B=\left(\begin{array}{ccccc}
 1 &1 &\ldots  &1&0\\ 
 \alpha_1 &\alpha_2 &\ldots & \alpha_{n+1-k}&0\\ 
 \vdots &\vdots &\ddots&\vdots&\vdots \\ 
 \alpha_1^{n-k-1} &\alpha_2^{n-k-1} &\ldots & \alpha^{n-k-1}_{n+1-k}&0\\ 
 \alpha_1^{n-k} &\alpha_2^{n-k} &\ldots & \alpha^{n-k}_{n+1-k}&-1\\
 \alpha_1^{n+1-k} &\alpha_2^{n+1-k} &\ldots & \alpha^{n+1-k}_{n+1-k}&\delta\\
\end{array}\right).\end{equation*} Let $A_i$ be the coefficient of $x^i$ in the polynomial $\mbox{det}(B(x))$, where
\begin{equation*}\label{eq11}
B(x)=\left(\begin{array}{ccccc}
 1 &1 &\ldots  &1&1\\ 
 \alpha_1 &\alpha_2 &\ldots & \alpha_{n+1-k}&x\\ 
 \vdots &\vdots &\ddots&\vdots&\vdots \\ 
 \alpha_1^{n+1-k} &\alpha_2^{n+1-k} &\ldots & \alpha^{n+1-k}_{n+1-k}&x^{n+1-k}\\
\end{array}\right).\end{equation*}
Then $\mbox{det}(B)=\delta A_{n+1-k}-A_{n-k}$. Note that $$\sum_{i=0}^{n+1-k}A_ix^i=\prod_{1\le i<j\le n+1-k}(\alpha_j-\alpha_i) \prod_{i=1}^{n+1-k}(x-\alpha_i).$$
Thus $A_{n+1-k}\neq 0$, $A_{n-k}=-A_{n+1-k}\sum_{i=1}^{n+1-k}\alpha_i$ and $\mbox{det}(B)\neq 0$ if and only if $\sum_{i=1}^{n+1-k}\alpha_i\neq -\delta.$ 

(2) When ${\bg}=(f_{\pi}(\alpha_1), \ldots, f_{\pi}(\alpha_n),\delta)$,  it suffices to prove that the submatrix consisting of any $(n+2-k)$ columns of the following matrix is nonsingular: 
\begin{equation}\label{eqHg2}
\bigg( \frac{H_{n+1-k,\infty}}{\bg}\bigg)=\left(\begin{array}{ccccc}
 w_1\alpha_1^0 &w_2\alpha_2^0 &\ldots  &w_n\alpha^0_{n}&0\\ 
 w_1\alpha_1^1 &w_2\alpha_2^1 &\ldots & w_n\alpha^1_{n}&0\\ 
 \vdots &\vdots &\ddots&\vdots&\vdots \\ 
 w_1\alpha_1^{n-k-1} &w_2\alpha_2^{n-k-1} &\ldots & w_n\alpha^{n-k-1}_{n}&0\\ w_1\alpha_1^{n-k} &w_2\alpha_2^{n-k} &\ldots & w_n\alpha^{n-k}_{n}&-1\\
 \frac{w_1}{\alpha_1-\pi} &\frac{w_1}{\alpha_1-\pi} &\ldots & \frac{w_n}{\alpha_n-\pi}&\delta\\
\end{array}\right).\end{equation}

If those columns from \eqref{eqHg2} do not containing the vector $(0,\ldots, 0, -1, \delta)^T$, then the result holds because each induced submatrix is  Vandermonde and $\alpha_i$ are distinct. Then we consider the other case, without loss of generality we need to compute the determinant of the   matrix:

\begin{equation*}\label{eq112}
B=\left(\begin{array}{ccccc}
\frac{1}{\alpha_1-\pi} &\frac{1}{\alpha_2-\pi}&\ldots & \frac{1}{\alpha_{n+1-k}-\pi}&\delta\\
 1 &1 &\ldots  &1&0\\ 
 \alpha_1 &\alpha_2 &\ldots & \alpha_{n+1-k}&0\\ 
 \vdots &\vdots &\ddots&\vdots&\vdots \\ 
 \alpha_1^{n-k-1} &\alpha_2^{n-k-1} &\ldots & \alpha^{n-k-1}_{n+1-k}&0\\ 
 \alpha_1^{n-k} &\alpha_2^{n-k} &\ldots & \alpha^{n-k}_{n+1-k}&-1\\
\end{array}\right).\end{equation*} 
Then \begin{equation*}
\prod_{i=1}^{n+1-k}\frac{1}{\alpha_i-\pi}\mbox{det}(B)=\left|\begin{array}{ccccc}
1 &1 &\ldots  &1&\delta\\ 
 \alpha_1 &\alpha_2 &\ldots & \alpha_{n+1-k}&\pi\delta\\ 
 \vdots &\vdots &\ddots&\vdots&\vdots \\ 
 \alpha_1^{n-k} &\alpha_2^{n-k} &\ldots & \alpha^{n-k}_{n+1-k}&\pi^{n-k}\delta\\
  \alpha_1^{n+1-k} &\alpha_2^{n+1-k} &\ldots & \alpha^{n+1-k}_{n+1-k}&\pi^{n+1-k}\delta-1\\
\end{array}\right|.\end{equation*} 
Let $\prod_{i=1}^{n+1-k}\frac{1}{\alpha_i-\pi}\mbox{det}(B)=\delta\mbox{det}(B_1)+\mbox{det}(B_2)$, where $$B_1=\left(\begin{array}{ccccc}
1  &\ldots  &1&1\\
 \alpha_1 &\ldots & \alpha_{n+1-k}&\pi\\ 
 \vdots &\ddots&\vdots&\vdots \\ 
 \alpha_1^{n-k}  &\ldots & \alpha^{n-k}_{n+1-k}&\pi^{n-k}\\
  \alpha_1^{n+1-k}  &\ldots & \alpha^{n+1-k}_{n+1-k}&\pi^{n+1-k}\\
\end{array}\right) \mbox{ and }B_2=\left(\begin{array}{ccccc}
1  &\ldots  &1&0\\ 
 \alpha_1  &\ldots & \alpha_{n+1-k}&0\\ 
 \vdots &\ddots&\vdots&\vdots \\ 
 \alpha_1^{n-k} &\ldots & \alpha^{n-k}_{n+1-k}&0\\
  \alpha_1^{n+1-k}  &\ldots & \alpha^{n+1-k}_{n+1-k}&-1\\
\end{array}\right).$$
Note that $$\mbox{det}(B_1)=\prod_{1\le i<j\le n+1-k}(\alpha_j-\alpha_i) \prod_{i=1}^{n+1-k}(\pi-\alpha_i)$$ and $$\mbox{det}(B_2)=-\prod_{1\le i<j\le n+1-k}(\alpha_j-\alpha_i) .$$
Therefore $\mbox{det}(B) \neq 0$ if and only if $ \delta\cdot\prod_{i=1}^{n+1-k}(\pi-\alpha_i)\neq  1$.
This completes the proof. 
%
%
%
\end{proof}

\begin{remark}{\rm By Corollary \ref{cor13} and Theorem \ref{thm14} (1), we have  $G_{k,\infty} {\bg}^T=(0, \ldots, 0, 1, \delta+a)^T$ with $a=\sum_{i=1}^n\alpha_i$. By  Lemma \ref{lem11}, the Roth-Lempel  code generated by the matrix in \eqref{eq8} (replace $\delta$ by $\delta+a$ )  is  MDS  if and only if the set $\{\alpha_1, \ldots, \alpha_{n}\}$ is an $(n, k-1, \delta+a)$-set in $\gf(q)$, 
 namely in the set $\{\alpha_1, \ldots, \alpha_n\}$  no $(k-1)$ elements  sum to $\delta+a$. On  the other hand, that also means the set $\{\alpha_1, \ldots, \alpha_{n}\}$ is an $(n, n+1-k, -\delta)$-set in $\gf(q)$, 
which is consistent with the condition  in  Theorem \ref{thm14} (1).


}
\end{remark}

\begin{remark}{\rm In Theorem 14 (2) the condition $\delta \notin T_{\pi,n+1-k}$ may be easily met in some cases. 
For example,  $\delta=0$ or 
if the set $\{\alpha_1, \ldots, \alpha_n\}$ properly is contained  to some subfield $\gf(s)$ of $\gf(q)$ and $\pi \in \gf(s)\backslash \{\alpha_1, \ldots, \alpha_n\}$, and $\delta\in \gf(q)\backslash \gf(s)$, then the condition holds.

}
\end{remark}

By Corollary \ref{cor10}, there are only two kinds of deep holes of 
$\mathcal{C}_{n-k}({\bf a}, {\bf w})$. Hence we have the following conjecture and  the reader is cordially invited to attack it.

\begin{conjecture}\label{conj} {\rm Let $q$ be a power of a prime,  $(q-1)/2\le n-k$, and $k<n\le q$. For a vector ${\bf u}=(u_1, u_2, \ldots,u_{n+1})\in \gf(q)^{n+1}$, the extended code $\overline{\mathcal{C}_k({\bf a}, {\bf 1},\infty)}({\bf u})$ is MDS if and only if  $\rho(\mathcal{C}_{n+1-k}({\bf a}, {\bf w},\infty))=k$ and  $\bu$ has one of the following forms: 

(1) $\bu=(\alpha_1f(\alpha_1), \ldots, \alpha_nf(\alpha_n),\delta)$ with $-\delta\notin S_{n+1-k}$, where 
 $f(x)$ was given in Theorem \ref{thm14} (1) and   $S_{n+1-k}$ was defined  in \eqref{eqs1}.
 

(2) $\bu=(f_{\pi}(\alpha_1), \ldots, f_{\pi}(\alpha_n),\delta)$ with $\delta\notin T_{\pi, n+1-k}$, where  $f_{\pi}(x)$ was given in Theorem \ref{thm14} (2) and  $T_{\pi, n+1-k}$ was defined  in \eqref{eqs2}.


}
\end{conjecture}

\section{The covering radii of MDS codes}\label{sec-special3} 

As an application of Theorem \ref{thm6}, in this section we will determine the covering radii of some  MDS codes. This will further demonstrate the usefulness and importance of Theorem \ref{thm6}. 

\subsection{Covering radii of projective Reed-Solomon codes}

In this subsection, we will determine the covering radii of some projective Reed-Solomon codes. 
Determining the covering radius of a given linear code over a finite field is generally a challenging task. The covering radius is a basic geometric parameter of a code and has been widely investigated for distinguished families of error-correcting codes, including Reed-Solomon codes and Reed-Muller codes, see \cite{BGP}. For an $[n,k]$ MDS code $\C$, it is known that the covering radius $\rho(\C)$ is either $n-k$ or $n-k-1$ \cite{BGP}. 
In \cite{ZWK}  it is mentioned that the covering radius of a $q$-ary $[n, k]$ Reed-Solomon code with  length $n < q + 1$ is known to be $n - k$.

In the case $n = q + 1$, these $q$-ary $[n, k]$ Reed-Solomon codes are referred to as {\it projective Reed-Solomon codes}, denoted as $\mbox{PRS}(k)$. Now let's consider the  projective Reed-Solomon codes. By Lemma \ref{lem3}, we have  $\mbox{PRS}(k)^{\bot}=\mbox{PRS}(q+1-k)$ as in this case  $w_i=-1$ in \eqref{eq4} for all $i$.  This result is derived from the fact that $\prod_{1\leq j\leq n,j\neq i}(\alpha_i-\alpha_j)=-1$, which is precisely the derivative of the polynomial $x^q-x$ at the element $\alpha_i$.  It was informed in \cite{ZWK} that for $k \in \{1, q, q + 1\}$ the covering radius of $\mbox{PRS}(k)$ is  $q - k+1$.  In the case of $k = q - 1$, the covering radius of $\mbox{PRS}(k)$ is $1$. However, for $2 \leq k \leq q - 2$, the covering radius of $\mbox{PRS}(k)$ is currently open and conjectured as follows.  
 \begin{conjecture}\label{conj:1} {\rm \cite[Conjecture I.2 ]{ZWK} For $2\le k\le q-2$, the covering radius of $\mbox{PRS}(k)$ is:
 \begin{eqnarray*}
\begin{cases}  q-k+1 \mbox{~ if } q \mbox{ is even and } k\in \{2,q-2\},\\
q-k  \mbox{ ~~~~~~otherwise}.\\
   \end{cases}
\end{eqnarray*} 
 }
 \end{conjecture}

 The above conjecture is true if $k \ge\left \lfloor \frac{q-1}{2} \right \rfloor$ from the work of Seroussi and Roth \cite{SR}. The following theorem gives a partial answer to Conjecture \ref{conj:1}.



\begin{theorem}\label{thm:conj}
{\rm
(1) When $ k\in \{2,q-2\} $, then the covering radius of $\mbox{PRS}(k)$ is equal to $q-k+1$  if and only if $q$  is even.

(2) When $3\le k \le q-3$, if the covering radius $\rho(\mbox{PRS}(k))=q-k$, then the set $\gf(q)$ is not a $(q,k,\delta)$-set in $\gf(q)$ for any $\delta\in \gf(q)$, namely for each $\delta\in \gf(q)$, there exist $k$ distinct elements of $\gf(q)$ whose sum is exactly $\delta$.

}
\end{theorem}
\begin{proof}
(1)  If $k=2$, then $\gf(q)$ is a $(q,2,0)$-set  in $\gf(q)$ if and only if $q$ is even.  Theorem \ref{thm612} tells us that $\rho(\mbox{PRS}(3)^{\bot})=\rho(\mbox{PRS}(q-2))=3=(q+1)-(q-2)$. 
If $k=q-2$, then $\gf(q)$ is a $(q,q-2,0)$-set  in $\gf(q)$ if and only if $q$ is even due to the identity $\sum_{x\in \gf(q)}x=0$. By Theorem \ref{thm612}, $\rho(\mbox{PRS}(q-1)^{\bot})=\rho(\mbox{PRS}(2))=q-1=(q+1)-2$. 


(2)  Suppose that $\rho(\mbox{PRS}(k))=q-k$.  By Theorem \ref{thm6},  $\overline{\text{PRS}(q+1-k)}(\mathbf{u})$ is not an MDS code. 
By Lemma \ref{lem11}, the set $\gf(q)$ is not a $(q,q-k,\delta)$-set in $\gf(q)$ for any $\delta\in \gf(q)$. So $\gf(q)$ is also not a $(q,k,-\delta)$-set.  The result follows from the fact that the argument holds for any arbitrary choice of $\delta$ in $\text{GF}(q)$.
This completes the proof.\end{proof}


From \cite{SR}, Conjecture \ref{conj:1} holds for $k \ge\left \lfloor \frac{q-1}{2} \right \rfloor$. Namely $\rho(\mbox{PRS}(k))=q-k$ when $\left \lfloor \frac{q-1}{2} \right \rfloor\le k<q-2$. We have the following corollary.

\begin{corollary}\label{cor:set}
{\rm
When $\left \lfloor \frac{q-1}{2} \right \rfloor\le k<q-2$, for any $\delta\in \gf(q)$ there exist $k$ distinct elements of $\gf(q)$ whose sum is equal to $\delta$.
}
\end{corollary}



 
 





\subsection{The covering radii of two families of MDS codes}
In this subsection, we will determine the covering radii of two families of MDS codes. 
We first 
introduce a class of MDS codes. 
Let $m$ be a positive integer, $q=2^m$, and $\alpha$ be a primitive element of $\gf(q^2)$. 
Define $\beta=\alpha^{q-1}$. Then $\beta$ is a $(q+1)$-th primitive root of unity in $\gf(q^2)$. Let $\m_{\beta^i}(x)$ denote 
the minimal polynomial of $\beta^i$ over $\gf(q)$. Clearly, $\m_{\beta^0}(x)=x-1$ and 
$$ 
\m_{\beta^i}(x)=(x-\beta^i)(x-\beta^{q+1-i})=(x-\beta^i)(x-\beta^{-i}) 
$$  
for all $1 \leq i \leq q$. 
For each $u$ with $1 \leq u \leq q/2$, define 
\begin{eqnarray}
g_u(x)=\m_{\beta^u}(x) \cdots \m_{\beta^{q/2}}(x).  
\end{eqnarray}
Let $\C_u$ be the cyclic code of length $q+1$ over $\gf(q)$ with generator polynomial $g_u(x)$. As pointed in \cite{SD}, here the code $\C_u$ and its dual code are both not Reed-Solomon codes. We have the following well-known result \cite{MS,SD}. 

\begin{lemma}\label{thm-sdjoin1}
{\rm \cite[Theorem 7]{SD}
Let $m \geq 2$. For each $u$ with $1 \leq u \leq q/2$, $\C_u$ is a $[q+1, 2u-1, q-2u+3]$ MDS cyclic code. 
} 
\end{lemma}

The next lemma was proved in \cite{SD}. 
 
 \begin{lemma}\label{thm-sdjoint2}
 {\rm  \cite[Corollary 17]{SD} Let $m \geq 2$. Then 
$\overline{\C_2}(\bone)$ is a $[q+2, 3, q]$ MDS code over $\gf(q)$ with weight enumerator 
$$ 
1 + \frac{(q+2)(q^2-1)}{2} z^q + \frac{q(q-1)^2}{2} z^{q+2}. 
$$ 
} 
\end{lemma} 

Combining Theorem \ref{thm6} and Lemma \ref{thm-sdjoint2} yields the following result. 

\begin{theorem}\label{thm-new1}
{\rm 
The covering radius $\rho(\C_2^\perp)=3$ and $\bone$ is a deep hole of $\C_2^\perp$. 
} 
\end{theorem}

The next lemma was proved in \cite{SD}.

\begin{lemma}\label{thm-sdjoint18}
{\rm \cite[Theorem 49]{SD} 
Let $m \geq 2$. Then $\overline{\C_{q/2}}(\bone)$ is a $[q+2, q-1, 4]$ MDS code over $\gf(q)$.
} 
\end{lemma} 

Combining Theorem \ref{thm6} and Lemma \ref{thm-sdjoint18} yields the following result. 

\begin{theorem}\label{thm-new2}
{\rm 
The covering radius $\rho(\C_{q/2}^\perp)=q-1$ and $\bone$ is a deep hole of $\C_{q/2}^\perp$. 
} 
\end{theorem}

\section{Summary and concluding remarks}\label{sec-final}

The main contributions of this paper are summarized as follows. 

\begin{itemize}
\item  For an MDS $[n,k]$ code $\C$, we proved that its extended code $\overline{\C}(\bu)$ for a vector $\bu$ 
 is still MDS if and only if  $\rho(\mathcal{C}^{\bot})=k$ and and $\bu$ is a deep hole of the dual code $\C^\perp$ (see Theorem \ref{thm6}). 
 This is the main result of this paper. 



\item With the help of known results on deep holes of the GRS codes, we determined all vectors $\bu$ such that 
the extended code $\overline{\C}(\bu)$ of the GRS code $\C$ is still MDS in some cases (see Corollary \ref{cor10}).  

\item With the help of known results on deep holes of the GRS codes, we determined some vectors $\bu$ such that the extended code $\overline{\C}(\bu)$ of the EGRS code $\C$ is still MDS in some cases (see Theorem \ref{thm14}). 

\item We determined the covering radius of $\mathcal{C}_k({\bf a}, {\bf 1},\infty)^\perp$ under certain conditions (see Theorem \ref{thm612}). 
\item We proved that the covering radius of  $\mbox{PRS}(k)$ is equal to $q-k+1$ when $q$  is even and $ k\in \{2,q-2\} $ (see Theorem \ref{thm:conj}). 
This provides a partial answer to Conjecture \ref{conj:1}.

\item We settled the covering radii of two families of MDS cyclic codes (see Theorems \ref{thm-new1} and \ref{thm-new2}). 
\end{itemize}

It was proved in \cite{SDC} that every linear code $\C'$ of length $n+1$ with minimum distance 
at least 2 is permutation-equivalent to an extended code $\overline{\C}(\bu)$ for some vector $\bu$ and some 
linear code $\C$ of length $n$ over the same finite field. This shows the importance of the extending techniques in 
coding theory. However, the determination of the minimum distance of $\overline{\C}(\bu)$ is very hard in most cases, 
even if the weight distribution of $\C$ is known. For example, the minimum distance of $\overline{\C}(\bu)$ is 
either 3 or 4 when $\C$ is a Hamming code over $\gf(q)$ for $q>2$. But it is open for which $\bu$'s the minimum 
distance  of the extended nonbinary Hamming code is exactly 4 \cite{SDC}. This hardness may explain why there are very limited results on extended linear codes in the literature.    

Given an MDS $[n,k]$ code $\C$ over $\gf(q)$ and a vector $\bu \in \gf(q)^n$, $\overline{\C}(\bu)$ is an MDS code if and only $\rho(\mathcal{C}^{\bot})=k$ and $\bu$ is a deep hole of $\C^\perp$ according to Theorem \ref{thm6}. However, determining the  deep holes of MDS codes is a very hard problem. Hence, Open Problem 1 for MDS codes is very hard in general. For instance, it is open if $\overline{\C}(-\bone)$ is MDS for many MDS codes $\C$ in a family of MDS codes \cite{SD}.  
Determining the covering radii of MDS codes is also a hard problem in general.   
It would be good if the conjectures presented in this paper could be proved or disproved.

\end{document}